\documentclass[lettersize,journal]{IEEEtran}
\usepackage{amsfonts}
\usepackage{amsmath}
\usepackage{amssymb}
\usepackage{amsthm}
\usepackage{epsfig}
\usepackage{multicol}
\usepackage{booktabs}
\usepackage{lettrine}
\usepackage{float}
\usepackage{footmisc}
\usepackage{epstopdf}
\usepackage[compress]{cite}
\usepackage{times}
\usepackage{stfloats}
\usepackage{mathtools}
\usepackage{algorithm}
\usepackage{algorithmic}
\usepackage{multirow}
\usepackage{color}
\usepackage{subfigure}
\usepackage{diagbox}
\usepackage[caption=false,font=normalsize,labelfont=sf,textfont=sf]{subfig}
\usepackage{makecell}
\usepackage{colortbl}
\usepackage{xcolor}
\usepackage{array}
\usepackage{arydshln}

\newtheorem{assumption}{Assumption}

\newtheorem{theorem}{Theorem}
\newtheorem{remark}{Remark}
\newtheorem{lemma}{Lemma}
\newtheorem{proposition}{Proposition}
\newtheorem{definition}{Definition}
\newtheorem{corollary}{Corollary}

\allowdisplaybreaks[4]
\IfFileExists{tcilatex.tex}{\input{tcilatex}}{}
\providecommand{\func}[1]{\operatorname{#1}}
\begin{document}

\title{Data-Driven Critic-Free Policy Iteration for Continuous-Time
Linear Quadratic Regulation \thanks{
}}
\author{Jiacheng Wu, {Yang~Zhu}\IEEEauthorrefmark{1}, ~Hongye~Su,
\IEEEmembership{Senior Member,~IEEE}

\thanks{
Jiacheng Wu, Yang Zhu, and Hongye Su are with State Key Laboratory of Industrial Control Technology,
Institute of Cyber-Systems and Control, Zhejiang University, Hangzhou
310027, China (e-mail: jiachengwu@zju.edu.cn;
zhuyang88@zju.edu.cn;
hysu@iipc.zju.edu.cn;).}

\thanks{\textsuperscript{*}Corresponding author: Yang Zhu.}%

}
\maketitle

\begin{abstract}
For continuous-time linear quadratic regulation with unknown system matrices, data-driven off-policy policy iteration typically estimates the value matrix and the
improved feedback gain through a joint critic--actor regression. We show that
the critic is not needed in the policy-improvement step. The key is to anchor
the Riccati equation at a known stabilizing gain and express optimality as a
policy-space residual. An endpoint null-space projection then removes the
value-matrix term from the integral data equation. This yields a critic-free,
actor-only least-squares update computed directly from input-state data.
Under a verifiable projected rank condition, the resulting data equation is
equivalent to the policy-space residual equation, and each update coincides
with the Kleinman iteration. Thus, the stabilizing and convergence
properties of Kleinman iteration are retained without a critic regression. We further show that the conventional off-policy
full-rank condition decomposes into an endpoint critic rank condition and a
projected actor rank condition. The proposed method removes the rank
requirement needed for critic identification while retaining the one needed
for policy improvement. The repeated least-squares dimension is reduced from
$n(n+1)/2+mn$ to $mn$. Finally, comparative simulations validate the effectiveness of the proposed
algorithm.
\end{abstract}


\section{Introduction}
\IEEEPARstart{C}ONTINUOUS-time linear quadratic regulation is a fundamental problem in optimal control \cite{lewis2012optimal}. When the system matrices are known, the optimal feedback gain can be obtained by solving the algebraic Riccati equation (ARE). In many practical applications, however, accurate knowledge of the system matrices is unavailable, which has motivated extensive research on data-driven optimal control \cite{zhao2025data}, reinforcement learning (RL) \cite{cui2025learning}, and adaptive dynamic programming \cite{liu2025adaptive}.

For this problem, data-driven policy iteration (PI) provides a natural way to implement Riccati-based optimal control from data \cite{vamvoudakis2014online,alemzadeh2024data,chen2019reinforcement}. Starting from an admissible feedback gain, PI alternates between policy evaluation and policy improvement, and its model-based form is closely related to the classical Kleinman iteration. In the data-driven setting, Vrabie \emph{et al.}~\cite{vrabie2009adaptive} introduced an on-policy PI scheme based on integral RL, where the feedback gain is updated from closed-loop input-state data rather than by directly solving the ARE. This approach reduces the reliance on a complete system model, but its policy-improvement step still requires partial knowledge of the system dynamics. To further reduce model dependence, Jiang \emph{et al.}~\cite{jiang2012computational} developed an off-policy PI framework in which integral regression equations are constructed from trajectories generated by a behavior input, so that the unknown system matrices need not be explicitly identified. This framework has motivated a broad range of extensions and applications, including inverse optimal control \cite{LianBS2024BOOK}, permanent magnet synchronous motor \cite{zhao2025reinforcement}, singularly perturbed systems \cite{shen2024data}, and networked \cite{jiang2025off} or multi-agent systems~\cite{xu2026distributed,donge2022multiagent}. Nevertheless, the resulting regression treats the value matrix of the current policy and the improved feedback gain as simultaneous unknowns. As a result, each policy update still requires a joint critic--actor regression for estimating both the value matrix and the feedback gain.

The joint critic--actor structure raises a fundamental question: \textit{is explicit critic identification indispensable for policy improvement?} In conventional off-policy PI \cite{jiang2012computational}, the value matrix enters the integral Bellman equation through the policy-evaluation term, whereas the policy-improvement step ultimately aims to compute only the next feedback gain. Hence, identifying the value matrix may introduce more information than is actually needed for the actor update. This distinction has direct computational and data-informativity implications. Since the critic contains \(n(n+1)/2\) unknowns whereas the actor contains only \(mn\), the critic component can dominate the regression size and impose additional rank requirements, particularly for high-dimensional systems or limited data batches. Related policy-space viewpoints have appeared in model-based studies. Bu \emph{et al.}~\cite{bu2020policy} studied direct policy-gradient methods for continuous-time linear quadratic regulation by optimizing the cost over stabilizing feedback gains, shifting the emphasis from Riccati-based value-function computation to policy-space optimization. However, evaluating the policy gradient still requires the value matrix associated with the current policy. More recently, Sassano~\cite{sassano2024policy} introduced policy algebraic equations for \(H_{\infty}\) control, and Possieri and Sassano~\cite{possieri2026solving} derived a quadratic Policy Algebraic Riccati Equation for  linear quadratic regulation whose unknowns are only the entries of the feedback gain. These results show that critic-free algebraic characterizations of optimal policies are possible when the model is known.

However, these model-based policy-space results \cite{bu2020policy,sassano2024policy,possieri2026solving} do not resolve the critic problem in data-driven off-policy PI. In the off-policy integral regression setting, the available identities are obtained from behavior trajectories, and the unknown value matrix appears in endpoint terms of the data equation. Simply discarding these terms would destroy the equivalence with Riccati-based policy improvement. Therefore, a critic-free data-driven method must eliminate the value-matrix contribution from the finite-data equation while preserving enough information to recover the same policy-improvement gain.
Motivated by the above observations, this paper develops a data-driven critic-free off-policy PI method for continuous-time linear quadratic regulation. The main contributions are summarized as follows.

\begin{itemize}
\item We develop a data-driven critic-free off-policy PI formulation for continuous-time linear quadratic regulation. By combining an anchored policy-space Riccati equation with an endpoint null-space projection, the value-matrix term is eliminated exactly from the finite-data equation, leading to an actor-only regression.

\item We prove that the proposed critic-free update is not an approximation of policy improvement. Under a verifiable projected rank condition, it uniquely recovers the Kleinman update. Thus, the stabilizing and convergence properties of  PI are retained without critic identification.

\item We clarify the data-informativity burden hidden in conventional off-policy PI. Its full-rank condition is decomposed into a critic rank condition and an actor rank condition, showing that the proposed method removes the rank requirement for value-matrix identification and reduces the repeated regression dimension from \(n(n+1)/2+mn\) to \(mn\).
\end{itemize}

\noindent\textbf{Notation:}
Throughout this article, $\mathbb R^n$ and $\mathbb R^{m\times n}$ denote
the $n$-dimensional Euclidean space and the set of $m\times n$ real
matrices, respectively. The symbol $I_n$ denotes the $n\times n$ identity matrix. The zero scalar, vector,
or matrix of compatible dimension is denoted by $\mathbf 0$. For a matrix
$\mathbf A$, $\mathbf A^T$, $\operatorname{rank}(\mathbf A)$,
$\ker(\mathbf A)$, $\operatorname{im}(\mathbf A)$, and
$\operatorname{Tr}(\mathbf A)$ denote its transpose, rank, kernel, image,
and trace, respectively. The notation $\operatorname{Re}(\lambda)$ denotes
the real part of a complex number $\lambda$.  
For a matrix $\mathbf A=[a_{ij}]\in\mathbb R^{m\times n}$,
$\operatorname{vec}(\mathbf A)$ denotes its column-stacking vectorization.
For a symmetric matrix $\mathbf P=[p_{ij}]\in\mathbb R^{n\times n}$, define
$
\operatorname{vecs}(\mathbf P)
=
[p_{11},2p_{12},\ldots,2p_{1n},p_{22},2p_{23},\ldots,p_{nn}]^T .
$
For a symmetric matrix $\mathbf S=[s_{ij}]\in\mathbb R^{n\times n}$, define
$
\operatorname{vech}(\mathbf S)
=
[s_{11},s_{12},\ldots,s_{1n},s_{22},s_{23},\ldots,s_{nn}]^T .
$
For vectors or matrices with compatible dimensions,
$\operatorname{col}(\cdot)$ denotes their column concatenation. For a
differentiable matrix-valued mapping $\mathcal M$, $D\mathcal M(K)[E]$
denotes the Fr\'echet derivative of $\mathcal M$ at $K$ in the direction
$E$.

\section{PRELIMINARIES AND PROBLEM STATEMENT}

Consider the continuous-time linear system
\begin{equation}
\dot{x}(t)=\mathbf{A}x(t)+\mathbf{B}u(t),  \label{1}
\end{equation}
where $x(t)\in \mathbb{R}^{n}$ and $u(t)\in \mathbb{R}^{m}$ denote the state
and control input, respectively. The constant matrices
$\mathbf{A}\in \mathbb{R}^{n\times n}$ and
$\mathbf{B}\in \mathbb{R}^{n\times m}$ are unknown. The performance index is
\begin{equation}
J(x_{0},u)=\int_{0}^{\infty }\left[
x^{T}(t)\mathbf{Q}x(t)+u^{T}(t)\mathbf{R}u(t)\right] \mathrm{d}t,  \label{2}
\end{equation}
where $\mathbf{Q}=\mathbf{Q}^{T}\geq \mathbf{0}$ and
$\mathbf{R}=\mathbf{R}^{T}>\mathbf{0}$ are known.

\begin{assumption}
The pair $(\mathbf{A},\mathbf{B})$ is stabilizable, the pair
$(\mathbf{A},\mathbf{Q}^{1/2})$ is detectable, and a gain
$\mathbf{F}\in \mathbb{R}^{m\times n}$ such that
$\mathbf{A}-\mathbf{B}\mathbf{F}$ is Hurwitz is available.
\end{assumption}

Under Assumption 1, the algebraic Riccati equation (ARE)
\begin{equation}
\mathbf{A}^{T}\mathbf{P}+\mathbf{P}\mathbf{A}
-\mathbf{P}\mathbf{B}\mathbf{R}^{-1}\mathbf{B}^{T}\mathbf{P}
+\mathbf{Q}=\mathbf{0}  \label{3}
\end{equation}
has a unique stabilizing solution
$\mathbf{P}^{\ast }=\mathbf{P}^{\ast T}\geq \mathbf{0}$. The corresponding
optimal policy is
\begin{equation}
u^{\ast }(t)=-\mathbf{K}^{\ast }x(t)
=-\mathbf{R}^{-1}\mathbf{B}^{T}\mathbf{P}^{\ast }x(t).  \label{4}
\end{equation}

Because $\mathbf{A}$ and $\mathbf{B}$ are unavailable, the ARE \eqref{3} cannot be
solved directly. Conventional off-policy PI avoids explicit
system identification by estimating the value matrix $\mathbf{P}_{k}$ and
the improved gain $\mathbf{K}_{k+1}$ simultaneously
\cite{jiang2012computational}. For an arbitrary behavior input $u(t)$,
\eqref{1} can be written as
\begin{equation}
\dot{x}(t)=(\mathbf{A}-\mathbf{B}\mathbf{K}_{k})x(t)
+\mathbf{B}\left[u(t)+\mathbf{K}_{k}x(t)\right].  \label{5}
\end{equation}
Let $\mathbf{P}_{k}$ be the value matrix associated with the stabilizing
policy $\mathbf{K}_{k}$, and define
$\mathbf{Q}_{k}\triangleq
\mathbf{Q}+\mathbf{K}_{k}^{T}\mathbf{R}\mathbf{K}_{k}$ and
$\mathbf{K}_{k+1}\triangleq \mathbf{R}^{-1}\mathbf{B}^{T}\mathbf{P}_{k}$.
Differentiating $x^{T}(t)\mathbf{P}_{k}x(t)$ along \eqref{5} yields
\begin{eqnarray}
&&x^{T}(t_{j+1})\mathbf{P}_{k}x(t_{j+1})
-x^{T}(t_{j})\mathbf{P}_{k}x(t_{j}) \notag \\
&=&-\int_{t_{j}}^{t_{j+1}}x^{T}(\tau )\mathbf{Q}_{k}x(\tau )\mathrm{d}\tau
\notag \\
&&\quad
+2\int_{t_{j}}^{t_{j+1}}
\left[u(\tau )+\mathbf{K}_{k}x(\tau )\right]^{T}
\mathbf{R}\mathbf{K}_{k+1}x(\tau )\mathrm{d}\tau .  \label{6}
\end{eqnarray}
This identity holds for any behavior input. In this article, the data are
collected under
$u(t)=-\mathbf{F}x(t)+v(t)$, where $v(t)\in \mathbb{R}^{m}$ is an exploration
signal and $\mathbf{F}$ is the stabilizing gain specified in Assumption 1.

Given $0\leq t_{0}<t_{1}<\cdots <t_{N}$, define
$\bar{x}(t)\triangleq \func{vech}(x(t)x^{T}(t))$ and
$r\triangleq n(n+1)/2$. Then
\begin{eqnarray*}
\boldsymbol{\delta }_{xx} &\triangleq &\left[
\begin{array}{c}
\bar{x}^{T}(t_{1})-\bar{x}^{T}(t_{0}) \\
\vdots \\
\bar{x}^{T}(t_{N})-\bar{x}^{T}(t_{N-1})
\end{array}
\right], \\
\mathbf{I}_{xx} &\triangleq &\left[
\begin{array}{c}
\displaystyle\int_{t_{0}}^{t_{1}}x^{T}(\tau )\otimes x^{T}(\tau )\mathrm{d}\tau \\
\vdots \\
\displaystyle\int_{t_{N-1}}^{t_{N}}x^{T}(\tau )\otimes x^{T}(\tau )\mathrm{d}\tau
\end{array}
\right], \\
\mathbf{I}_{xu} &\triangleq &\left[
\begin{array}{c}
\displaystyle\int_{t_{0}}^{t_{1}}x^{T}(\tau )\otimes u^{T}(\tau )\mathrm{d}\tau \\
\vdots \\
\displaystyle\int_{t_{N-1}}^{t_{N}}x^{T}(\tau )\otimes u^{T}(\tau )\mathrm{d}\tau
\end{array}
\right].
\end{eqnarray*}
Stacking \eqref{6} gives
\begin{equation}
\boldsymbol{\Theta }_{k}
\left[
\begin{array}{c}
\mathrm{vecs}(\mathbf{P}_{k}) \\
\func{vec}(\mathbf{K}_{k+1})
\end{array}
\right]
=-\mathbf{I}_{xx}\func{vec}(\mathbf{Q}_{k}),  \label{7}
\end{equation}
where
\begin{equation*}
\boldsymbol{\Theta }_{k}\triangleq
\left[
\begin{array}{cc}
\boldsymbol{\delta }_{xx} &
-2\mathbf{I}_{xx}(\mathbf{I}_{n}\otimes \mathbf{K}_{k}^{T}\mathbf{R})
-2\mathbf{I}_{xu}(\mathbf{I}_{n}\otimes \mathbf{R})
\end{array}
\right].
\end{equation*}
\begin{assumption}
The collected data satisfy
\textbf{$\func{rank}(\boldsymbol{\Theta }_{k})=r+mn$}.
\end{assumption}

If $\func{rank}(\boldsymbol{\Theta }_{k})=r+mn$, then
$\mathbf{P}_{k}$ and $\mathbf{K}_{k+1}$ are uniquely determined by joint critic-actor regression \eqref{7}. Thus,
although conventional off-policy PI is model-free, it still
requires an $(r+mn)$-dimensional joint critic--actor regression \eqref{7}. The
objective of the following development is to remove the critic variable from
the policy-improvement equation while retaining a data-driven
implementation.

\section{DATA-DRIVEN CRITIC-FREE RL}

The development proceeds through an anchored policy-space Riccati equation
and an endpoint null-space projection. The former recasts the optimality
condition as a residual equation in the policy variable, while the latter
removes the critic-dependent endpoint term from the data equation.

We first introduce an anchored Riccati equation in the policy space. Let $\mathbf{A}_{F}\triangleq \mathbf{A}-\mathbf{B}\mathbf{F}$. Since $\mathbf A_F$ is Hurwitz, the Lyapunov operator
$\mathcal L_F(\mathbf P)\triangleq
\mathbf A_F^T\mathbf P+\mathbf P\mathbf A_F
$
is nonsingular on the space of symmetric matrices. This allows every policy
$\mathbf K\in\mathbb R^{m\times n}$ to be lifted to a unique symmetric
matrix through a fixed stable Lyapunov equation. This motivates the following definition.

\begin{definition}
For each $\mathbf K\in\mathbb R^{m\times n}$, let
$\mathbf P_F(\mathbf K)=\mathbf P_F^T(\mathbf K)$ be the unique solution of
\begin{equation}
\mathbf A_F^T\mathbf P_F(\mathbf K)
+\mathbf P_F(\mathbf K)\mathbf A_F
=
\mathbf H_F(\mathbf K),
\label{8}
\end{equation}
where
\begin{equation}
\mathbf H_F(\mathbf K)\triangleq
(\mathbf K-\mathbf F)^T\mathbf R(\mathbf K-\mathbf F)
-\mathbf Q-\mathbf F^T\mathbf R\mathbf F .
\label{9}
\end{equation}
Define the anchored policy-space residual
\begin{equation}
\mathbf G_F(\mathbf K)\triangleq
\mathbf R\mathbf K-\mathbf B^T\mathbf P_F(\mathbf K).
\label{10}
\end{equation}
The equation $\mathbf G_F(\mathbf K)=\mathbf 0$ is called the anchored policy-space Riccati equation associated with the
anchor $\mathbf F$.
\end{definition}

\begin{lemma}
The equation $\mathbf{G}_{F}(\mathbf{K})=\mathbf{0}$ is equivalent to the ARE \eqref{3} together with
$\mathbf{K}=\mathbf{R}^{-1}\mathbf{B}^{T}\mathbf{P}$. Consequently,
$\mathbf{K}^{\ast}$ is its unique stabilizing solution.
\end{lemma}

\begin{proof}
Let $\mathbf{G}_{F}(\mathbf{K})=\mathbf{0}$ and set
$\mathbf{P}=\mathbf{P}_{F}(\mathbf{K})$. Then
$\mathbf{B}^{T}\mathbf{P}=\mathbf{R}\mathbf{K}$. Since $\mathbf{P}$ and
$\mathbf{R}$ are symmetric, it follows that
$\mathbf{P}\mathbf{B}=\mathbf{K}^{T}\mathbf{R}$. Substituting these
identities into \eqref{8} and expanding $\mathbf{A}_{F}$ and
$\mathbf{H}_{F}(\mathbf{K})$ gives
$\mathbf{A}^{T}\mathbf{P}+\mathbf{P}\mathbf{A}
-\mathbf{K}^{T}\mathbf{R}\mathbf{K}+\mathbf{Q}=\mathbf{0}$. Together with
$\mathbf{K}=\mathbf{R}^{-1}\mathbf{B}^{T}\mathbf{P}$, this equation is
exactly \eqref{3}. Conversely, let a symmetric $\mathbf{P}$ satisfy
\eqref{3} and set $\mathbf{K}=\mathbf{R}^{-1}\mathbf{B}^{T}\mathbf{P}$.
Direct substitution shows that $\mathbf{P}$ satisfies \eqref{8}. By
uniqueness of the Lyapunov solution,
$\mathbf{P}=\mathbf{P}_{F}(\mathbf{K})$, and therefore
$\mathbf{G}_{F}(\mathbf{K})=\mathbf{0}$. The final statement follows from
the uniqueness of the stabilizing solution of \eqref{3}.
\end{proof}

Lemma 1 shows that the optimal control gain can be recovered by solving the
policy-space equation $\mathbf G_F(\mathbf K)=\mathbf 0$. However,
$\mathbf G_F(\mathbf K)$ is not directly computable from data because it
contains the unknown term $\mathbf B^T\mathbf P_F(\mathbf K)$. We now derive
an anchored integral identity in which $\mathbf P_F(\mathbf K)$ appears only
through endpoint terms, and then eliminate these terms by an endpoint
null-space projection.
 Suppose
that the behavior input $u(t)=-\mathbf{F}x(t)+v(t)$ is applied on
$\mathcal{I}_{j}=[t_{j-1},t_{j}]$, $j=1,\ldots,N$, and define
\begin{eqnarray}
\mathbf{X}_{j} &\triangleq&
\int_{t_{j-1}}^{t_{j}}x(\tau )x^{T}(\tau )\,\mathrm{d}\tau,  \label{11} \\
\mathbf{Z}_{j} &\triangleq&
\int_{t_{j-1}}^{t_{j}}v(\tau )x^{T}(\tau )\,\mathrm{d}\tau,  \label{12} \\
\mathbf{D}_{j} &\triangleq&
x(t_{j})x^{T}(t_{j})-x(t_{j-1})x^{T}(t_{j-1}). \label{13}
\end{eqnarray}
The resulting trajectory satisfies
$\dot{x}(t)=\mathbf{A}_{F}x(t)+\mathbf{B}v(t)$. For any fixed
$\mathbf K$, differentiating
$x^{T}(t)\mathbf{P}_{F}(\mathbf K)x(t)$ along this trajectory gives
\begin{eqnarray}
&&x^{T}(t_{j})\mathbf{P}_{F}(\mathbf{K})x(t_{j})
-x^{T}(t_{j-1})\mathbf{P}_{F}(\mathbf{K})x(t_{j-1})  \notag \\
&=&\int_{t_{j-1}}^{t_{j}}
x^{T}(\tau )\mathbf{H}_{F}(\mathbf{K})x(\tau )\,\mathrm{d}\tau \notag \\
&&+
2\int_{t_{j-1}}^{t_{j}}
v^{T}(\tau )\mathbf{B}^{T}\mathbf{P}_{F}(\mathbf{K})x(\tau )\,\mathrm{d}\tau .
\label{14}
\end{eqnarray}
Using
$\mathbf{B}^{T}\mathbf{P}_{F}(\mathbf{K})
=\mathbf{R}\mathbf{K}-\mathbf{G}_{F}(\mathbf{K})$, \eqref{14} becomes
\begin{eqnarray}
\!\!\!\!&&\int_{t_{j-1}}^{t_{j}}
\!\!x^{T}(\tau )\mathbf{H}_{F}(\mathbf{K})x(\tau )\,\mathrm{d}\tau
+2\int_{t_{j-1}}^{t_{j}}
\!\!\!\!v^{T}(\tau )\mathbf{R}\mathbf{K}x(\tau )\,\mathrm{d}\tau \notag \\
&=&x^{T}(t_{j})\mathbf{P}_{F}(\mathbf{K})x(t_{j})
-x^{T}(t_{j-1})\mathbf{P}_{F}(\mathbf{K})x(t_{j-1}) \notag \\
&&+
2\int_{t_{j-1}}^{t_{j}}
v^{T}(\tau )\mathbf{G}_{F}(\mathbf{K})x(\tau )\,\mathrm{d}\tau .
\label{15}
\end{eqnarray}
In \eqref{15}, the endpoint term is the only term that contains the unknown
matrix $\mathbf{P}_{F}(\mathbf K)$. This term can be eliminated by projecting
the interval identities onto the null space generated by the endpoint data.

Define
$\boldsymbol{\delta}_{j}\triangleq
\bar{x}(t_{j})-\bar{x}(t_{j-1})=\func{vech}(\mathbf D_j)$ and
\[
\boldsymbol{\mathcal{D}}\triangleq
\begin{bmatrix}
\boldsymbol{\delta}_{1} & \cdots & \boldsymbol{\delta}_{N}
\end{bmatrix}
\in \mathbb{R}^{r\times N}.
\]
Let $\mathbf W\in\mathbb R^{N\times q}$ be an orthonormal basis of
$\ker(\boldsymbol{\mathcal D})$, where
$q=N-\func{rank}(\boldsymbol{\mathcal D})$, and write its $\ell$th column as
$w_{\ell}=\func{col}(w_{1\ell},\ldots,w_{N\ell})$. Since
$w_{\ell}\in\ker(\boldsymbol{\mathcal D})$,
\[
\sum_{j=1}^{N}w_{j\ell}\boldsymbol{\delta}_{j}=\mathbf 0,
\qquad
\sum_{j=1}^{N}w_{j\ell}\mathbf D_{j}=\mathbf 0 .
\]
After multiplying \eqref{15} by $w_{j\ell}$ and summing over
$j=1,\ldots,N$, the endpoint term becomes
\begin{eqnarray}
\!\!\!\!&&\sum_{j=1}^{N}w_{j\ell}
\Bigl[
x^{T}(t_{j})\mathbf{P}_{F}(\mathbf{K})x(t_{j})
-x^{T}(t_{j-1})\mathbf{P}_{F}(\mathbf{K})x(t_{j-1})
\Bigr] \notag \\
\!\!\!\!&=&
\func{Tr}\left(
\mathbf{P}_{F}(\mathbf K)
\sum_{j=1}^{N}w_{j\ell}\mathbf D_j
\right)=\textbf{0}. \label{16}
\end{eqnarray}
\begin{remark}
The endpoint projection does not approximate the term involving
$\mathbf P_F(\mathbf K)$. It eliminates this term exactly. Indeed, every
column of $\mathbf W$ defines a linear combination of sampling intervals for
which the quadratic endpoint increment vanishes for any symmetric matrix.
Therefore, the critic-dependent endpoint component is removed independently
of the particular candidate policy $\mathbf K$.
\end{remark}
Hence, for each $\ell=1,\ldots,q$,
\begin{eqnarray}
&&\sum_{j=1}^{N}w_{j\ell}
\int_{t_{j-1}}^{t_{j}}
x^{T}(\tau )\mathbf{H}_{F}(\mathbf{K})x(\tau )\,\mathrm{d}\tau \notag \\
&&+
2\sum_{j=1}^{N}w_{j\ell}
\int_{t_{j-1}}^{t_{j}}
v^{T}(\tau )\mathbf{R}\mathbf{K}x(\tau )\,\mathrm{d}\tau \notag \\
&=&
2\sum_{j=1}^{N}w_{j\ell}
\int_{t_{j-1}}^{t_{j}}
v^{T}(\tau )\mathbf{G}_{F}(\mathbf{K})x(\tau )\,\mathrm{d}\tau .
\label{17}
\end{eqnarray}
Define the projected policy-space data map by
\begin{eqnarray}
[\boldsymbol{\Phi}_{F}(\mathbf{K})]_{\ell}
&\triangleq&
\sum_{j=1}^{N}w_{j\ell}
\int_{t_{j-1}}^{t_{j}}
x^{T}(\tau )\mathbf{H}_{F}(\mathbf{K})x(\tau )\,\mathrm{d}\tau \notag \\
&&+
2\sum_{j=1}^{N}w_{j\ell}
\int_{t_{j-1}}^{t_{j}}
v^{T}(\tau )\mathbf{R}\mathbf{K}x(\tau )\,\mathrm{d}\tau .
\label{18}
\end{eqnarray}
Then \eqref{17} gives
\begin{equation}
[\boldsymbol{\Phi}_{F}(\mathbf{K})]_{\ell}
=
2\sum_{j=1}^{N}w_{j\ell}
\int_{t_{j-1}}^{t_{j}}
v^{T}(\tau )\mathbf{G}_{F}(\mathbf{K})x(\tau )\,\mathrm{d}\tau .
\label{19}
\end{equation}
Let
\[
\bar{\mathbf Z}_{\ell}\triangleq
\sum_{j=1}^{N}w_{j\ell}\mathbf Z_j,
\qquad
\boldsymbol{\mathcal Z}\triangleq
\begin{bmatrix}
\func{vec}(\bar{\mathbf Z}_{1})^{T}\\
\vdots\\
\func{vec}(\bar{\mathbf Z}_{q})^{T}
\end{bmatrix}
\in\mathbb R^{q\times mn}.
\]
Stacking \eqref{19} over $\ell=1,\ldots,q$ yields the residual factorization
\begin{equation}
\boldsymbol{\Phi}_{F}(\mathbf{K})
=
2\boldsymbol{\mathcal Z}\func{vec}(\mathbf{G}_{F}(\mathbf{K})).
\label{20}
\end{equation}
Thus, when the projected data matrix has full column rank, the projected data
equation is equivalent to the policy-space residual equation. Before imposing a rank condition on the projected matrix
$\boldsymbol{\mathcal Z}$, we verify that this condition is independent of
the particular orthonormal basis chosen for $\ker(\boldsymbol{\mathcal D})$.
\begin{proposition}
Let $\mathbf W$ and $\widetilde{\mathbf W}$ be two orthonormal bases of
$\ker(\boldsymbol{\mathcal D})$. Let
$\boldsymbol{\Phi}_F$, $\boldsymbol{\mathcal Z}$ and
$\widetilde{\boldsymbol{\Phi}}_F$, $\widetilde{\boldsymbol{\mathcal Z}}$ be
the projected quantities induced by these two bases. Then,
$
\widetilde{\boldsymbol{\Phi}}_F(\mathbf K)=\mathbf T^T
\boldsymbol{\Phi}_F(\mathbf K)$ and
$\widetilde{\boldsymbol{\mathcal Z}}=\mathbf T^T\boldsymbol{\mathcal Z}
$
for some orthogonal matrix $\mathbf T$. Consequently,
$\func{rank}(\boldsymbol{\mathcal Z})$ and the zero set of
$\boldsymbol{\Phi}_F(\mathbf K)=\mathbf 0$ are independent of the chosen
orthonormal basis of $\ker(\boldsymbol{\mathcal D})$.
\end{proposition}

\begin{proof}
Since $\mathbf W$ and $\widetilde{\mathbf W}$ are orthonormal bases of the
same subspace, there exists an orthogonal matrix $\mathbf T$ such that
$\widetilde{\mathbf W}=\mathbf W\mathbf T$. The definitions of the projected
interval sums therefore give
$\widetilde{\boldsymbol{\Phi}}_F(\mathbf K)=\mathbf T^T
\boldsymbol{\Phi}_F(\mathbf K)$ and
$\widetilde{\boldsymbol{\mathcal Z}}=\mathbf T^T\boldsymbol{\mathcal Z}$.
Since $\mathbf T$ is nonsingular, rank and nullity are preserved.
\end{proof}
\begin{assumption}
The projected data matrix $\boldsymbol{\mathcal Z}$ has full column rank, i.e.,
$\mathrm{rank}(\boldsymbol{\mathcal Z})=mn$.
\end{assumption}
\begin{remark}
Assumption~3 is a projected finite-data informativity condition. It is
directly checkable from data and is not guaranteed by persistent excitation
of $v(t)$ alone, because $\boldsymbol{\mathcal Z}$ also depends on the
state response and the endpoint projection. Since
$\boldsymbol{\mathcal Z}\in\mathbb R^{q\times mn}$, it necessarily requires
$q\ge mn$.
\end{remark}
\begin{theorem}
Under Assumptions 1 and 3, the data-driven policy-space equation
$\boldsymbol{\Phi}_{F}(\mathbf K)=\mathbf 0$ and the residual equation
$\mathbf G_F(\mathbf K)=\mathbf 0$ have the same solution set. Consequently,
$\boldsymbol{\Phi}_{F}(\mathbf K)=\mathbf 0$ admits $\mathbf K^\ast$ as its
unique stabilizing solution.
\end{theorem}

\begin{proof}
Let $\boldsymbol{\Phi}_{F}(\mathbf K)=\mathbf 0$. From \eqref{20},
$
\boldsymbol{\mathcal Z}
\func{vec}(\mathbf G_F(\mathbf K))=\mathbf 0 .
$
By Assumption 3, $\boldsymbol{\mathcal Z}$ has full column rank. Hence
$\func{vec}(\mathbf G_F(\mathbf K))=\mathbf 0$, or equivalently
$\mathbf G_F(\mathbf K)=\mathbf 0$. Conversely, if
$\mathbf G_F(\mathbf K)=\mathbf 0$, then \eqref{20} directly gives
$\boldsymbol{\Phi}_{F}(\mathbf K)=\mathbf 0$. The two equations therefore
have the same solution set. The uniqueness of the stabilizing solution
then follows from Lemma 1 and the uniqueness of the stabilizing solution of the
ARE \eqref{3}.
\end{proof}

We now construct an actor-only policy-improvement step from the projected
equation. For a fixed $\mathbf K$, the Fr\'echet derivative of
$\boldsymbol{\Phi}_{F}$ in the direction
$\mathbf E\in\mathbb R^{m\times n}$ is obtained by differentiating
$\mathbf H_F(\mathbf K)$ and the term linear in $\mathbf K$. Since
\begin{equation}
D\mathbf{H}_{F}(\mathbf{K})[\mathbf{E}]
=
\mathbf{E}^{T}\mathbf{R}(\mathbf{K}-\mathbf{F})
+
(\mathbf{K}-\mathbf{F})^{T}\mathbf{R}\mathbf{E},
\label{21}
\end{equation}
we obtain, for $\ell=1,\ldots,q$,
\begin{eqnarray}
&&D[\boldsymbol{\Phi}_{F}(\mathbf{K})]_{\ell}[\mathbf{E}]
\notag\\
&=&
\!\!\!\!2\sum_{j=1}^{N}w_{j\ell}
\!\!\int_{t_{j-1}}^{t_{j}}
\!\!\!\!\left[(\mathbf K-\mathbf F)x(\tau)+v(\tau)\right]^{T}
\!\!\mathbf R\mathbf E x(\tau)\,\mathrm d\tau .
\label{22}
\end{eqnarray}
At the $k$th iteration, we apply one Newton step to
$\boldsymbol{\Phi}_{F}(\mathbf K)=\mathbf 0$:
\begin{equation}
D\boldsymbol{\Phi}_{F}(\mathbf{K}_{k})
[\mathbf{K}_{k+1}-\mathbf{K}_{k}]
=
-\boldsymbol{\Phi}_{F}(\mathbf{K}_{k}).
\label{23}
\end{equation}
Using the definition of $\boldsymbol{\Phi}_{F}$ and collecting the terms that
multiply $\mathbf K_{k+1}$, \eqref{23} gives
\begin{eqnarray}
\!\!\!\!\!\!&&2\sum_{j=1}^{N}w_{j\ell}
\int_{t_{j-1}}^{t_{j}}
\left[(\mathbf K_k-\mathbf F)x(\tau)+v(\tau)\right]^{T}
\mathbf R\mathbf K_{k+1}x(\tau)\,\mathrm d\tau
\notag\\
\!\!\!\!\!\!&=&
\sum_{j=1}^{N}w_{j\ell}
\int_{t_{j-1}}^{t_{j}}
x^{T}(\tau)\mathbf Q_kx(\tau)\,\mathrm d\tau,  \ell=1,\ldots,q .
\label{24}
\end{eqnarray}
Equivalently,
\begin{equation}
\boldsymbol{\Gamma}_{k}\func{vec}(\mathbf{K}_{k+1})
=
\mathbf y_k ,
\label{25}
\end{equation}
where
$\bar{\mathbf X}_{\ell}\triangleq
\sum_{j=1}^{N}w_{j\ell}\mathbf X_j$,
and
\begin{eqnarray*}
\lbrack \boldsymbol{\Gamma }_{k}]_{\ell ,:} &\triangleq &2\left[ \func{vec}%
\left( \mathbf{R}\left[ (\mathbf{K}_{k}-\mathbf{F})\bar{\mathbf{X}}_{\ell }+%
\bar{\mathbf{Z}}_{\ell }\right] \right) \right] ^{T}, \\
\lbrack \mathbf{y}_{k}]_{\ell} &\triangleq &\sum_{j=1}^{N}w_{j\ell
}\int_{t_{j-1}}^{t_{j}}x^{T}(\tau )\mathbf{Q}_{k}x(\tau )\,\mathrm{d}\tau
\text{.}
\end{eqnarray*}%
Thus, the policy-improvement step \eqref{25} is computed from projected input-state
data, with $\mathbf K_{k+1}$ as the only unknown.

\begin{theorem}
Suppose Assumptions 1 and 3 hold. For every stabilizing
$\mathbf K_k$, the matrix $\boldsymbol{\Gamma}_k$ in \eqref{25} has full
column rank, i.e.,
$\func{rank}(\boldsymbol{\Gamma}_k)=mn$. Moreover, \eqref{25} admits a
unique solution, and this solution is
$\mathbf K_{k+1}=\mathbf R^{-1}\mathbf B^T\mathbf P_k$,
where $\mathbf P_k=\mathbf P_k^T$ is the unique solution of
\begin{equation}
(\mathbf A-\mathbf B\mathbf K_k)^T\mathbf P_k
+\mathbf P_k(\mathbf A-\mathbf B\mathbf K_k)
+\mathbf Q_k
=
\mathbf 0 .
\label{205}
\end{equation}
\end{theorem}

\begin{proof}
We first verify that the claimed policy-improvement gain satisfies
\eqref{25}. Since $\mathbf K_k$ is stabilizing, the Lyapunov equation
\eqref{205} admits a unique symmetric solution $\mathbf P_k$. Set
$\mathbf K_{k+1}=\mathbf R^{-1}\mathbf B^T\mathbf P_k$. Applying
\eqref{6} on $[t_{j-1},t_j]$ under
$u(t)=-\mathbf F x(t)+v(t)$ gives
\begin{eqnarray}
&&x^{T}(t_{j})\mathbf{P}_{k}x(t_{j})
-x^{T}(t_{j-1})\mathbf{P}_{k}x(t_{j-1}) \notag\\
&=&
-\int_{t_{j-1}}^{t_{j}}
x^{T}(\tau)\mathbf Q_kx(\tau)\,\mathrm d\tau \notag\\
&&\!\!\!\!\!\!\!\!+
2\int_{t_{j-1}}^{t_{j}}
\!\!\!\!\left[(\mathbf K_k-\mathbf F)x(\tau)+v(\tau)\right]^T
\mathbf R\mathbf K_{k+1}x(\tau)\,\mathrm d\tau .
\label{26}
\end{eqnarray}
Multiplying \eqref{26} by $w_{j\ell}$ and summing over
$j=1,\ldots,N$ eliminates the endpoint term because
$\sum_{j=1}^{N}w_{j\ell}\mathbf D_j=\mathbf 0$. By the definitions of
$\boldsymbol{\Gamma}_k$ and $\mathbf y_k$, the resulting identity is exactly
the $\ell$th row of \eqref{25}. Thus, this $\mathbf K_{k+1}$ is a
solution of \eqref{25}.

It remains to prove uniqueness. Let
$\mathbf E\in\mathbb R^{m\times n}$ satisfy
$\boldsymbol{\Gamma}_k\func{vec}(\mathbf E)=\mathbf 0$. By \eqref{22}, this
is equivalent to
$D\boldsymbol{\Phi}_F(\mathbf K_k)[\mathbf E]=\mathbf 0$. Let
$\boldsymbol{\Pi}=\boldsymbol{\Pi}^T$ be the unique solution of
\begin{equation}
\mathbf A_F^T\boldsymbol{\Pi}
+\boldsymbol{\Pi}\mathbf A_F
=
\mathbf E^T\mathbf R(\mathbf K_k-\mathbf F)
+
(\mathbf K_k-\mathbf F)^T\mathbf R\mathbf E .
\label{27}
\end{equation}
This equation is well posed because $\mathbf A_F$ is Hurwitz; moreover,
$\boldsymbol{\Pi}=D\mathbf P_F(\mathbf K_k)[\mathbf E]$. Differentiating
\eqref{20} at $\mathbf K_k$ in the direction $\mathbf E$ yields
\begin{equation}
D\boldsymbol{\Phi}_F(\mathbf K_k)[\mathbf E]
=
2\boldsymbol{\mathcal Z}
\func{vec}\left(\mathbf R\mathbf E-\mathbf B^T\boldsymbol{\Pi}\right).
\label{28}
\end{equation}
Since $D\boldsymbol{\Phi}_F(\mathbf K_k)[\mathbf E]=\mathbf 0$ and
$\boldsymbol{\mathcal Z}$ has full column rank by Assumption 3,
$
\mathbf R\mathbf E-\mathbf B^T\boldsymbol{\Pi}=\mathbf 0 .
$
By symmetry of $\mathbf R$ and $\boldsymbol{\Pi}$, this also gives
$\mathbf E^T\mathbf R=\boldsymbol{\Pi}\mathbf B$. Substituting these two
identities into \eqref{27} gives
\begin{equation}
\mathbf A_F^T\boldsymbol{\Pi}
+\boldsymbol{\Pi}\mathbf A_F
=
\boldsymbol{\Pi}\mathbf B(\mathbf K_k-\mathbf F)
+
(\mathbf K_k-\mathbf F)^T\mathbf B^T\boldsymbol{\Pi}.
\label{208}
\end{equation}
Using $\mathbf A_F=\mathbf A-\mathbf B\mathbf F$, \eqref{208} is equivalent
to
$
(\mathbf A-\mathbf B\mathbf K_k)^T\boldsymbol{\Pi}
+
\boldsymbol{\Pi}(\mathbf A-\mathbf B\mathbf K_k)
=
\mathbf 0 .
$
Since $\mathbf K_k$ is stabilizing, this Lyapunov equation has the
unique symmetric solution $\boldsymbol{\Pi}=\mathbf 0$. Hence
$\mathbf R\mathbf E=\mathbf 0$, and $\mathbf R>\textbf{0}$ implies
$\mathbf E=\mathbf 0$. Therefore, $\ker(\boldsymbol{\Gamma}_k)=\{\mathbf 0\}$.
Because $\boldsymbol{\Gamma}_k$ has $mn$ columns, it has full column rank.
Consequently, the solution of \eqref{25} is unique, and it is the Kleinman
gain $\mathbf K_{k+1}=\mathbf R^{-1}\mathbf B^T\mathbf P_k$.
\end{proof}

\begin{corollary}
Under the conditions of Theorem\ 2, Algorithm 1 generates the Kleinman
sequence. Hence, every $\mathbf{K}_{k}$ is stabilizing,
$\mathbf{P}^{\ast}\leq \mathbf{P}_{k+1}\leq \mathbf{P}_{k}$, and
$\lim_{k\rightarrow \infty}\mathbf{K}_{k}=\mathbf{K}^{\ast}$.
\end{corollary}
\begin{proof}
By Theorem 2, the unique solution of \eqref{25} is the Kleinman
policy-improvement gain
$\mathbf K_{k+1}=\mathbf R^{-1}\mathbf B^T\mathbf P_k$. Since
$\mathbf K_0=\mathbf F$ is stabilizing by Assumption 1, the standard
properties of Kleinman iteration \cite{kleinman1968iterative} imply that all subsequent gains are
stabilizing, the sequence $\{\mathbf P_k\}$ is monotonically nonincreasing
and bounded below by $\mathbf P^\ast$, and
$\mathbf K_k\to\mathbf K^\ast$.
\end{proof}
\begin{algorithm}[t]
\caption{Critic-Free Policy Iteration}
\renewcommand{\algorithmicrequire}{\textbf{Input:}}
\renewcommand{\algorithmicensure}{\textbf{Output:}}
\begin{algorithmic}[1]
\REQUIRE Stabilizing gain $\mathbf{F}$, exploration signal $v$, and
$\varepsilon>0$.
\ENSURE Optimal gain $\mathbf K^{\ast}$.
\REPEAT
\STATE Apply $u(t)=-\mathbf{F}x(t)+v(t)$ and augment the data set with
$\mathbf{X}_{j}$, $\mathbf{Z}_{j}$, and $\mathbf{D}_{j}$ from
\eqref{11}--\eqref{13}.
\STATE Construct $\boldsymbol{\mathcal{D}}$ and an orthonormal basis
$\mathbf{W}$ of $\ker(\boldsymbol{\mathcal{D}})$.
\STATE Construct $\bar{\mathbf{X}}_{\ell}$ and
$\bar{\mathbf{Z}}_{\ell}$.
\UNTIL{$\func{rank}(\boldsymbol{\mathcal{Z}})=mn$}
\STATE Set $\mathbf{K}_{0}\gets\mathbf{F}$ and $k\gets0$.
\REPEAT
\STATE Compute
$\func{vec}(\mathbf{K}_{k+1})\gets
(\boldsymbol{\Gamma}_{k}^{T}\boldsymbol{\Gamma}_{k})^{-1}
\boldsymbol{\Gamma}_{k}^{T}\mathbf{y}_{k}.$
\STATE Set $k\gets k+1$.
\UNTIL{$\Vert\mathbf{K}_{k}-\mathbf{K}_{k-1}\Vert\leq\varepsilon$}
\end{algorithmic}
\end{algorithm}

\subsection{Rank Decomposition of the Conventional Off-Policy Regression}

We next decompose the full-rank condition of the conventional off-policy
regression in \eqref{7}. This decomposition shows that the endpoint
projection separates the critic endpoint component from the actor component,
and that Assumption~3 is precisely the full-rank requirement on the projected
actor component.

Under the behavior input $u(t)=-\mathbf F x(t)+v(t)$, write
$
\boldsymbol{\Theta}_k=
\begin{bmatrix}
\boldsymbol{\delta}_{xx} & \mathbf C_k
\end{bmatrix},
$
where $\mathbf C_k\in\mathbb R^{N\times mn}$ denotes the actor block in
$\boldsymbol{\Theta}_k$. Its $j$th row is
\begin{equation}
[\mathbf C_k]_{j,:}
=
-2\left[
\func{vec}\left(
\mathbf R\left[
(\mathbf K_k-\mathbf F)\mathbf X_j+\mathbf Z_j
\right]\right)
\right]^T .
\label{30}
\end{equation}
Since $\boldsymbol{\delta}_{xx}=\boldsymbol{\mathcal D}^T$ and the columns
of $\mathbf W$ form an orthonormal basis of
$\ker(\boldsymbol{\mathcal D})$, we have
$
\mathbf W^T\boldsymbol{\delta}_{xx}=\mathbf 0 .
$
Moreover, by the definitions of $\bar{\mathbf X}_{\ell}$,
$\bar{\mathbf Z}_{\ell}$, and $\boldsymbol{\Gamma}_k$,
$
\mathbf W^T\mathbf C_k=-\boldsymbol{\Gamma}_k .
$

\begin{theorem}
Let $\mathbf K_k$ be stabilizing and define
$d\triangleq\func{rank}(\boldsymbol{\mathcal D})$. Then the following
statements hold.

\begin{enumerate}
\item[(i)] The regression matrices satisfy
\begin{equation}
\func{rank}(\boldsymbol{\Theta}_k)
=
d+\func{rank}(\boldsymbol{\Gamma}_k)
=
d+\func{rank}(\boldsymbol{\mathcal Z}).
\label{31}
\end{equation}

\item[(ii)] If
$\func{rank}(\boldsymbol{\Theta}_k)=r+mn$, then
\begin{equation}
\func{rank}(\boldsymbol{\mathcal D})=r,
\qquad
\func{rank}(\boldsymbol{\mathcal Z})=mn .
\label{32}
\end{equation}

\item[(iii)] If
$\func{rank}(\boldsymbol{\mathcal D})=r$ and
$\func{rank}(\boldsymbol{\mathcal Z})=mn$, then
\begin{equation}
\func{rank}(\boldsymbol{\Theta}_k)=r+mn .
\label{33}
\end{equation}
\end{enumerate}
\end{theorem}

\begin{proof}
We first prove statement (i). Since
$\boldsymbol{\delta}_{xx}=\boldsymbol{\mathcal D}^T$, we have
$\func{rank}(\boldsymbol{\delta}_{xx})=d$. Let
$\mathbf U\in\mathbb R^{N\times d}$ have orthonormal columns spanning
$\func{im}(\boldsymbol{\delta}_{xx})$. Since the columns of $\mathbf W$ form
an orthonormal basis of $\ker(\boldsymbol{\delta}_{xx}^T)$, the matrix
$
\mathbf T\triangleq
\begin{bmatrix}
\mathbf U & \mathbf W
\end{bmatrix}
$
is orthogonal. Premultiplication by $\mathbf T^T$ therefore preserves rank,
and
\begin{equation}
\mathbf T^T\boldsymbol{\Theta}_k
=
\begin{bmatrix}
\mathbf U^T\boldsymbol{\delta}_{xx} & \mathbf U^T\mathbf C_k\\
\mathbf 0 & \mathbf W^T\mathbf C_k
\end{bmatrix}.
\label{34}
\end{equation}
Because $\mathbf U^T\boldsymbol{\delta}_{xx}$ has full row rank $d$, the
block upper-triangular form in \eqref{34} gives
\begin{equation}
\func{rank}(\boldsymbol{\Theta}_k)
=
d+\func{rank}(\mathbf W^T\mathbf C_k)
=
d+\func{rank}(\boldsymbol{\Gamma}_k).
\label{35}
\end{equation}

It remains to prove that
$\func{rank}(\boldsymbol{\Gamma}_k)=
\func{rank}(\boldsymbol{\mathcal Z})$. Define
$\mathbf J_G(\mathbf K)\in\mathbb R^{mn\times mn}$ by
\begin{equation}
\func{vec}\left(D\mathbf G_F(\mathbf K)[\mathbf E]\right)
=
\mathbf J_G(\mathbf K)\func{vec}(\mathbf E).
\label{36}
\end{equation}
Differentiating \eqref{20} at $\mathbf K_k$ and using \eqref{22} yields
$
\boldsymbol{\Gamma}_k
=
2\boldsymbol{\mathcal Z}\mathbf J_G(\mathbf K_k).
$
We now show that $\mathbf J_G(\mathbf K_k)$ is nonsingular. Suppose
$D\mathbf G_F(\mathbf K_k)[\mathbf E]=\mathbf 0$, and set
$\boldsymbol{\Pi}\triangleq D\mathbf P_F(\mathbf K_k)[\mathbf E]$. Then
\begin{equation}
\mathbf A_F^T\boldsymbol{\Pi}
+\boldsymbol{\Pi}\mathbf A_F
=
\mathbf E^T\mathbf R(\mathbf K_k-\mathbf F)
+
(\mathbf K_k-\mathbf F)^T\mathbf R\mathbf E .
\label{37}
\end{equation}
Moreover,
$D\mathbf G_F(\mathbf K_k)[\mathbf E]=\mathbf 0$ implies
$\mathbf R\mathbf E=\mathbf B^T\boldsymbol{\Pi}$ and, by symmetry,
$\mathbf E^T\mathbf R=\boldsymbol{\Pi}\mathbf B$. Substituting these
relations into \eqref{37} gives
\begin{equation}
(\mathbf A-\mathbf B\mathbf K_k)^T\boldsymbol{\Pi}
+
\boldsymbol{\Pi}(\mathbf A-\mathbf B\mathbf K_k)
=
\mathbf 0 .
\label{38}
\end{equation}
Since $\mathbf K_k$ is stabilizing, \eqref{38} has the unique symmetric
solution $\boldsymbol{\Pi}=\mathbf 0$. Hence
$\mathbf R\mathbf E=\mathbf 0$, and $\mathbf R>0$ implies
$\mathbf E=\mathbf 0$. Therefore
$\ker(\mathbf J_G(\mathbf K_k))=\{\mathbf 0\}$, so
$\mathbf J_G(\mathbf K_k)$ is nonsingular. This proves
$\func{rank}(\boldsymbol{\Gamma}_k)=
\func{rank}(\boldsymbol{\mathcal Z})$, and statement (i) follows from
\eqref{35}.

For statement (ii), suppose
$\func{rank}(\boldsymbol{\Theta}_k)=r+mn$. By statement (i),
\[
r+mn=d+\func{rank}(\boldsymbol{\mathcal Z}).
\]
Since $d\leq r$ and
$\func{rank}(\boldsymbol{\mathcal Z})\leq mn$, this equality implies
\[
(r-d)+\bigl(mn-\func{rank}(\boldsymbol{\mathcal Z})\bigr)=\textbf{0} .
\]
Both terms are nonnegative; hence both must be zero. Thus
$d=r$ and $\func{rank}(\boldsymbol{\mathcal Z})=mn$, which proves
\eqref{32}. Statement (iii) follows directly by substituting
$d=r$ and $\func{rank}(\boldsymbol{\mathcal Z})=mn$ into \eqref{31}.
\end{proof}
\begin{corollary}
Suppose Assumptions~1 and~3 hold and
$\func{rank}(\boldsymbol{\mathcal D})<r$. Then, for every stabilizing
$\mathbf K_k$, the critic-free actor equation \eqref{25} uniquely determines
the Kleinman policy-improvement gain, whereas the conventional joint
regression \eqref{7} is rank deficient.
\end{corollary}

\begin{proof}
By Theorem~2, \eqref{25} has the unique solution
$\mathbf K_{k+1}=\mathbf R^{-1}\mathbf B^T\mathbf P_k$. Let
$d=\func{rank}(\boldsymbol{\mathcal D})<r$. Since Assumption~3 gives
$\func{rank}(\boldsymbol{\mathcal Z})=mn$, Theorem~3 yields
\[
\func{rank}(\boldsymbol{\Theta}_k)
=
d+\func{rank}(\boldsymbol{\mathcal Z})
=
d+mn
<r+mn .
\]
Thus, the conventional regression matrix in \eqref{7} does not have full
column rank and cannot uniquely determine the joint critic--actor unknown.
\end{proof}
\begin{remark}
Theorem~3 and Corollary~2 show that Assumption~3 is the projected actor
informativity condition hidden inside the conventional full-rank condition.
The conventional condition additionally requires the endpoint critic
component $\boldsymbol{\mathcal D}$ to have full row rank. Hence, endpoint
projection removes precisely the rank requirement needed for identifying the
critic, while retaining the rank requirement needed for identifying the
policy-improvement gain.
\end{remark}
\subsection{Computational Complexity Analysis}

We compare the per-iteration computational cost of the dense least-squares
solve used in the conventional off-policy regression and in the proposed
critic-free update. For a least-squares problem with $\eta$ equations and
$\Xi$ unknowns, where $\eta\geq \Xi$, a  dense solve has leading
complexity $\mathcal O(\eta\Xi^2)$ \cite{prashanth2014fast}.

In the conventional regression \eqref{7}, the unknown vector contains both
the critic variable $\mathrm{vecs}(\mathbf P_k)$ and the actor variable
$\func{vec}(\mathbf K_{k+1})$.  When the conventional full-rank condition holds, one
necessarily has $N\ge r+mn$, and the corresponding per-iteration
least-squares cost is $\mathcal O\bigl(N(r+mn)^2\bigr).$
By contrast, after the endpoint projection has been constructed, the
critic-free update \eqref{25} contains only the actor variable
$\func{vec}(\mathbf K_{k+1})\in\mathbb R^{mn}$. The projected equation has
$
q=N-\func{rank}(\boldsymbol{\mathcal D})
$
rows, and Assumption~3 requires $q\ge mn$. Hence, the per-iteration
least-squares cost of the critic-free update is
$
\mathcal O\bigl(q(mn)^2\bigr).
$
Thus, when both regressions are well posed on the same data batch, the ratio
between the conventional and critic-free least-squares costs is
\[
\frac{N}{q}
\left(
\frac{r+mn}{mn}
\right)^2
=
\frac{N}{q}
\left(
1+\frac{n+1}{2m}
\right)^2 .
\]
The saving becomes more pronounced when the number of states is large
relative to the number of inputs.

The construction of $\boldsymbol{\mathcal D}$, the null-space basis
$\mathbf W$, and the projected data matrices is performed only once after
data collection and is not repeated during the PI loop.
Moreover, if $\func{rank}(\boldsymbol{\mathcal D})<r$, the conventional
joint regression is rank deficient, whereas the critic-free update can still
be well posed under Assumption~3, as shown in Corollary~2.

\section{Illustrative Examples}

This section provides numerical examples to demonstrate the effectiveness and computational efficiency of the proposed critic-free PI method. We first compare it with existing off-policy methods on a benchmark system, and then examine its scalability over different system dimensions.

\textit{Example 1: Comparison with Existing Off-Policy Methods.} We compare the proposed critic-free Algorithm~1 with two existing
off-policy methods, namely the off-policy PI method in
\cite{jiang2012computational} and the off-policy Q-learning method in
\cite{li2019policy}. The turbocharged diesel engine with exhaust gas recirculation
studied in \cite{jiang2012computational} is considered. The system
parameters, weighting matrices, initial stabilizing gain, exploration
signal, and data-collection settings are chosen as in
\cite{jiang2012computational}.

The comparison results are shown in Fig.~1. The gain-error curves indicate
that Algorithm~1 and the off-policy PI method generate almost identical
gain trajectories and converge within five iterations. The off-policy
Q-learning method also converges to a comparable gain accuracy, but requires
more regression parameters and higher per-iteration computation.

The computational advantage of Algorithm~1 is due to its critic-free
structure. In this example, Algorithm~1 estimates only $mn=12$ actor
parameters, whereas the off-policy PI method estimates $r+mn=33$
critic--actor parameters and the off-policy Q-learning method estimates
$48$ unknown parameters. This reduction is also reflected in the timing
benchmark: Algorithm~1 takes $0.1888$ ms per iteration, compared with
$0.4419$ ms for the off-policy PI method and $0.8677$ ms for the
off-policy Q-learning method. Thus, Algorithm~1 achieves comparable gain
accuracy with the smallest regression dimension and the lowest
per-iteration computational cost. The singular-value plot confirms that all
three methods are evaluated under their respective full-rank data
conditions.
\begin{figure}[t]
\centering
\includegraphics[width=1.1\linewidth]{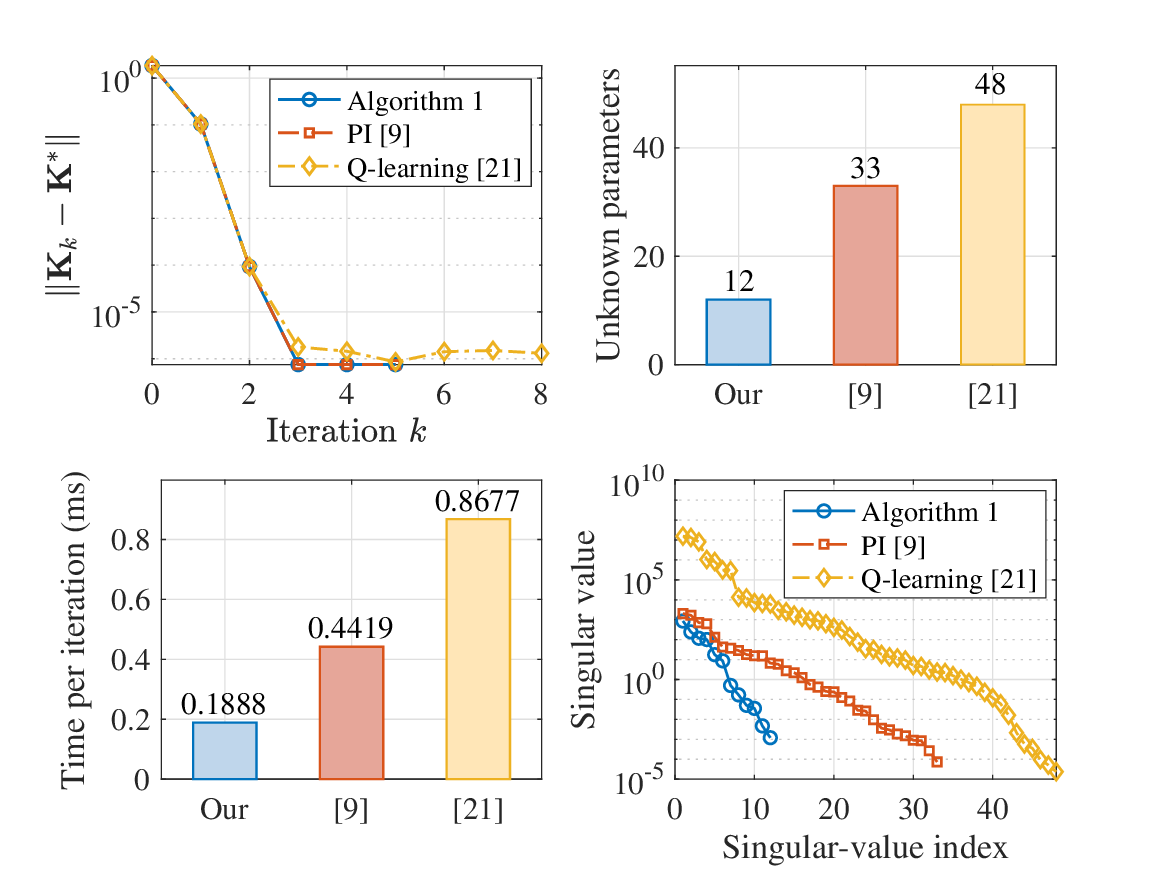}
\caption{Comparison with the off-policy PI method in \cite{jiang2012computational} and the off-policy Q-learning method in \cite{li2019policy}. Algorithm 1 attains comparable gain convergence with the smallest regression dimension and the lowest per-iteration computational cost under valid full-rank data conditions.}
\end{figure}

\textit{Example 2: Computational Scalability.}
We next evaluate the computational scalability of the proposed critic-free
update. The state dimension and input dimension are chosen as
$n=10,12,\ldots,28$ and $m=1,2,\ldots,10$, respectively. For each pair
$(n,m)$, system parameters are randomly generated. Specifically,
a random matrix $\mathbf{\widetilde A}$ is shifted to the open left-half plane by
setting
\[
\mathbf{A}=\mathbf{\widetilde A}-\left(\max_i \operatorname{Re}\lambda_i(\mathbf{\widetilde A})+0.5\right)I_n,
\]
and the input matrix $\mathbf{B}$ is obtained from a randomly generated matrix with
orthonormal columns, scaled by $1/\sqrt n$. The weighting matrices are
chosen as $\mathbf{Q}=I_n$ and $\mathbf{R}=I_m$. Since $\mathbf{A}$ is Hurwitz, the initial stabilizing
gain is set as $\mathbf{F}=\mathbf{0}$. All rank conditions required by the two
methods are checked before running the iterative updates.

Fig. 2 reports the runtime per iteration. The proposed
Algorithm~1 consistently requires less computation time than the conventional
off-policy PI method. This is consistent with the complexity analysis:
Algorithm~1 estimates only the $mn$ actor parameters, whereas the
conventional regression estimates both the critic and actor variables with
dimension $r+mn$. The runtime gap becomes more evident when the state
dimension is large relative to the input dimension.

\begin{figure}[t]
\centering
\includegraphics[width=0.9\linewidth]{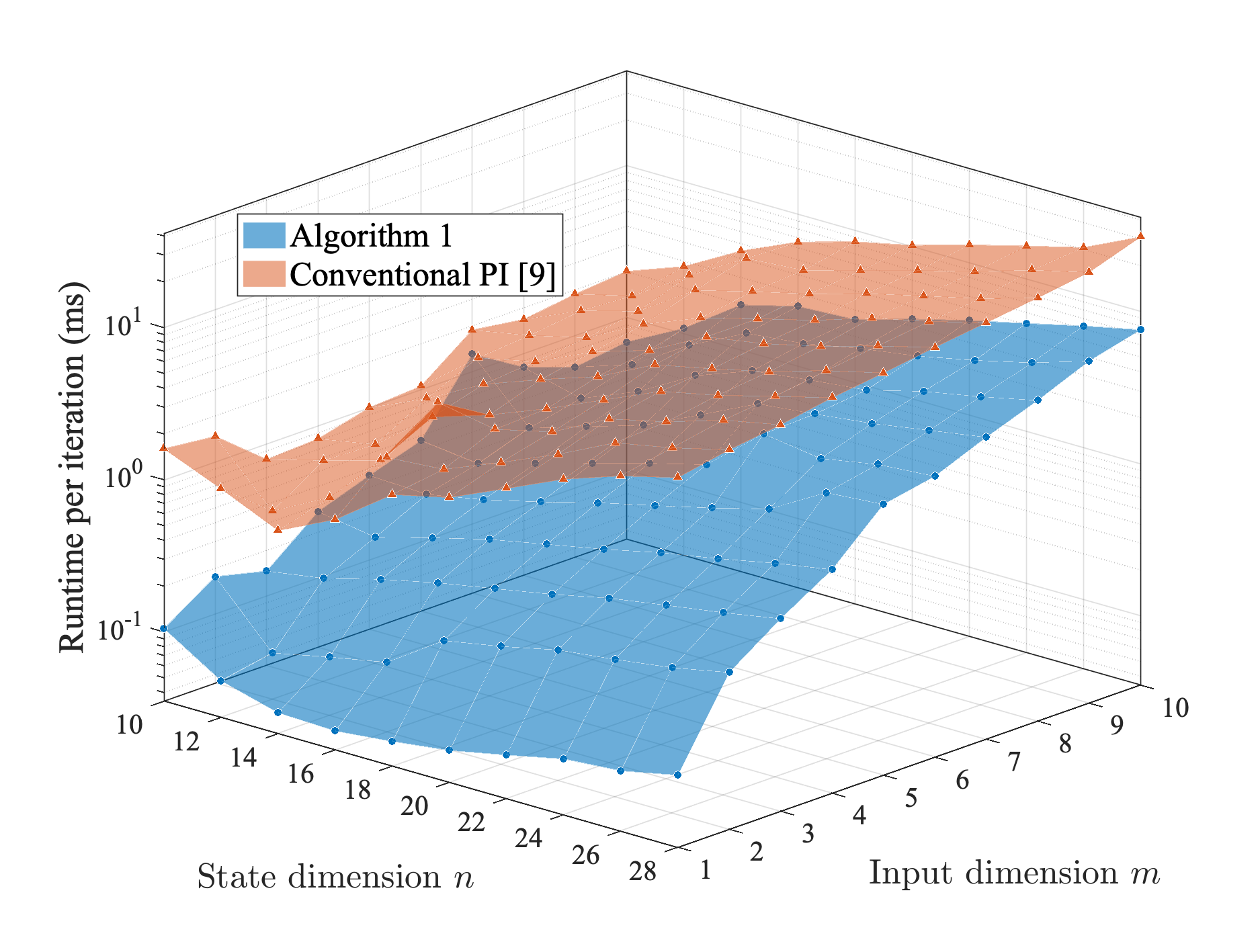}
\caption{Per-iteration runtime comparison between Algorithm~1 and the
conventional off-policy PI method in \cite{jiang2012computational}. Each
point corresponds to the runtime for a fixed pair $(n,m)$,
using the same system and data batch for both methods. The lower runtime of
Algorithm~1 reflects its reduced regression dimension from $r+mn$ to $mn$.}
\end{figure}
\section{Conclusion}

This paper showed that the critic variable in conventional off-policy PI is
not essential for computing the optimal control policy.
The anchored formulation and endpoint projection reveal that the data
informativity required for policy improvement can be separated from that
required for value-matrix identification. In particular, the projected actor
rank condition is sufficient to recover the Kleinman update, whereas the
conventional joint regression additionally requires endpoint critic
informativity. This separation explains both the reduction of the repeated
regression dimension from $n(n+1)/2+mn$ to $mn$ and the ability of the
critic-free update to remain well posed when the conventional regression is
rank deficient. Comparative simulations further confirmed that the proposed
implementation preserves the convergence behavior of off-policy PI while
reducing the per-iteration computational burden.

\bibliographystyle{IEEEtran}
\bibliography{shwj-workwjc}

\end{document}